%% file: ms.tex
\documentclass[letterpaper, 10 pt, conference]{ieeeconf}
\IEEEoverridecommandlockouts
\usepackage[utf8]{inputenc}
\usepackage{amsmath}
\usepackage{mathtools}                      % bigtimes
\usepackage{amssymb}
\usepackage{mathrsfs}
\usepackage{subcaption}
\usepackage[ruled]{algorithm2e}

\usepackage{todonotes}

\usepackage{mathrsfs}

\usepackage{amsthm}

\usepackage{graphicx}
\graphicspath{{./figures/}}

\input{macrosv1.tex}	% For bold symbols

\title{Underapproximation of Reach-Avoid Sets for Discrete-Time Stochastic Systems via Lagrangian Methods}
\author{Joseph D. Gleason$^\dagger$, Abraham P. Vinod$^\dagger$, Meeko. M. K. Oishi\thanks{This material is based upon work supported by the National Science
        Foundation under Grant Number IIS-1528047, CMMI-1254990 (Oishi, CAREER),
        CNS-1329878, and the Lighting Enabled 
    Systems and Applications Engineering Research Center (EEC-0812056). Any opinions, findings, and conclusions or recommendations expressed in this material are those of the authors and do not necessarily reflect the views of the National Science Foundation. 
    \newline \indent
    Joseph Gleason, Abraham Vinod, and Meeko Oishi are with Electrical and
Computer Engineering, University of New Mexico, Albuquerque, NM; e-mail:
gleasonj@unm.edu, aby.vinod@gmail.com, oishi@unm.edu (corresponding author)
    \newline \indent
    $\dagger$ These authors contributed equally to this work.}
}
\date{}

\DeclareMathOperator{\boundeddist}{\mathcal{E}}

\DeclareMathOperator{\rad}{\mathcal{D}}
\DeclareMathOperator{\reachf}{\mathscr{F}_1}
\DeclareMathOperator{\reach}{\mathscr{R}_1}
\DeclareMathOperator{\reachd}{\mathscr{R}_1}
\DeclareMathOperator{\viab}{\mathcal{V}}

\DeclareMathOperator{\minkdiff}{\ominus}
\def\ONE{\boldsymbol{1}}
\newtheorem{lemma}{Lemma}

\newtheorem{prop}{Proposition}
\newtheorem{theorem}{Theorem}
\newtheorem{corollary}{Corollary}
\newtheorem{prob}{Problem}
\newtheorem{innercustomthm}{Problem}
\newenvironment{customthm}[1]
  {\renewcommand\theinnercustomthm{#1}\innercustomthm}
  {\endinnercustomthm}

\begin{document}

\maketitle

\begin{abstract}
    We examine Lagrangian techniques for computing underapproximations of
    finite-time horizon, stochastic reach-avoid level-sets for discrete-time,
    nonlinear systems. We use the concept of reachability of a target tube in
    the control literature to define robust reach-avoid sets which are
    parameterized by the target set, safe set, and the set in which the
    disturbance is drawn from. We unify two existing Lagrangian approaches to
    compute these sets and establish that there exists an optimal control policy of the
    robust reach-avoid sets which is a Markov policy. Based on these results, we
    characterize the subset of the disturbance space whose corresponding robust
    reach-avoid set for the given target and safe set is a guaranteed
    underapproximation of the stochastic reach-avoid level-set of interest.
    The proposed approach dramatically improves the computational
    efficiency for obtaining an underapproximation of stochastic reach-avoid
    level-sets when compared to the traditional approaches based
    on gridding. 
    Our method, while conservative, does not rely on a grid, implying
    scalability as permitted by the known computational geometry constraints. We
    demonstrate the method on two examples: a simple two-dimensional integrator,
    and a space vehicle rendezvous-docking problem.
\end{abstract}

\section{Introduction}
\label{sec:introduction}

\input{introduction}

\section{Problem Statement}
\label{sec:problem-statement}

\input{problem-statement}

\section{Robust Reach-Avoid Set Computation}
\label{sec:lagrangian}

\input{lagrangian}

\section{Conservative Approximation of Stochastic Reach-Avoid Level-Set}
\label{sec:stochastic-approx}

\input{stochastic-approx}

\section{Examples}
\label{sec:examples}

\input{examples}

\section{Conclusion}
\label{sec:conlusion}

\input{conclusion}

\bibliographystyle{IEEEtran}
\bibliography{IEEEabrv,shortIEEE,ref}

\end{document}

%% file: macrosv1.tex
\usepackage{xspace}
\usepackage{amsmath,amssymb,amsfonts}

\newcommand{\bs}[1]{\ensuremath{\boldsymbol{#1}}}

\newcommand{\pci}[1]{\ensuremath{\bs{p}_{\boldsymbol{c}_{i}}}\xspace}

\def\QED{\ensuremath{\rule[0pt]{1.5ex}{1.5ex}}\xspace}

%% file: introduction.tex
Reach-avoid analysis is an established verification tool that provides formal
guarantees of both safety (via avoiding unsafe regions) and performance (via
reaching a target set).  It has been used in safety-critical or expensive systems, 
for example, with application to space systems
\cite{lesser2013_spacecraft}, aviation \cite{Tomlin2003,summers2011stochastic}, biomedical systems~\cite{maidens_2013}, and other domains~\cite{KariotoglouECC2011,ManganiniCYB2015,summers2010_verification}. 
The reach-avoid set is the set of initial states for which there exists control that enables the state to reach a target within some finite time horizon, while remaining within a safe set (avoiding
an unsafe set) for all instants in the time horizon.  In a probabilistic system, satisfaction of the reach-avoid objective is accomplished stochastically.  The {\em stochastic reach-avoid level-set} for a given likelihood is the set of states for which probabilistic success of the reach-avoid objective is assured with at least the given likelihood.

The theoretical framework for the probabilistic reach-avoid calculation is based on dynamic programming \cite{abate2008_reachability,summers2010_verification}, and, hence, is computationally infeasible for even moderate-sized systems due to the gridding of not only the state-space, but also of the input and disturbance spaces \cite{AbateHSCC2007}.  Recent work has focused on alternatives to dynamic programming, including    
approximate dynamic programming
\cite{KariotoglouECC2013, KariotoglouSCL2016, ManganiniCYB2015}, Gaussian
mixtures \cite{KariotoglouSCL2016}, particle filters
\cite{lesser2013_spacecraft, ManganiniCYB2015}, and convex chance-constrained
optimization \cite{lesser2013_spacecraft, KariotoglouECC2011}. These methods
have been applied to systems that are at most 10-dimensional, at high memory
and computational costs \cite{ManganiniCYB2015}.  Further, since an analytical expression
of the value function is not accessible, stochastic reach-avoid level-sets can
be computed only up to the accuracy of the gridding. 

We propose a method to compute an underapproximation of probabilistic reach-avoid sets via 
{\em robust reach-avoid sets}, the set of states assured to reach the target set and remain in the safe region despite {\em any} disturbance input.  
Robust reach-avoid sets can be theoretically posed as the solution to the \emph{reachability of a target tube} problem \cite{bertsekas1971minimax,kerrigan2001robust,rakovic2006_reach}, originally framed to compute reachable sets of discrete-time controlled systems with bounded disturbance sets. Motivated by the scalability of the Lagrangian method proposed in~\cite{maidens_2013, saintpierre1994_viability} for viability analysis in deterministic systems (that is, systems without a disturbance input but with a control input), we seek a similar approach to compute the robust reach-avoid sets via tractable set theoretic operations.
Lagrangian methods rely on computational geometry, whose scalability depends on
the representation and the operation used~\cite{guernic2009}, including
polyhedrons (implementable using Model Parametric Toolbox (MPT) \cite{MPT3}), support functions~\cite{leguernic2010}, and ellipsoids (implementable via the Ellipsoidal Toolbox~\cite{ellipsoidal}).

In this paper, we unify these two approaches to create an efficient algorithm for underapproximation of the stochastic reach-avoid set, and demonstrate our approach on practical examples.  Our main contributions are: a) synthesis of the approaches presented in \cite{maidens_2013,saintpierre1994_viability} and~\cite{bertsekas1971minimax,kerrigan2001robust,rakovic2006_reach} to compute the robust reach-avoid sets, b) sufficient conditions under which an optimal control policy for a given robust reach-avoid set is a Markov policy, and c) an algorithm to compute an underapproximation of the stochastic reach-avoid level-sets using the robust reach-avoid sets. 
Specifically, we establish the sufficient conditions under which an optimal control policy is comprised of universally measurable state-feedback laws. For these conditions, we characterize the subset of the disturbance space whose corresponding robust
reach-avoid set is a guaranteed underapproximation of the desired stochastic reach-avoid
level-set. Leveraging established
Lagrangian methods, we demonstrate that our approach dramatically reduces the computation time required for
computing a conservative underapproximation of the desired stochastic
reach-avoid level-set. Further, the Lagrangian methods does not rely on grids, freeing the underapproximated sets from any numerical artifacts arising due to the discretization. 

The remainder of the paper is as follows: Section \ref{sec:problem-statement}
describes the problem and the necessary notation. In Section \ref{sec:lagrangian},
we describe the relationship between the 
the recursion established
in~\cite{bertsekas1971minimax} for the robust reach-avoid set and the Lagrangian approach in~\cite{maidens_2013},
and establish the desired measurability properties of the
optimal controller.  We present an algorithm for underapproximation of stochastic
reach-avoid level-sets in Section \ref{sec:stochastic-approx}.  We demonstrate our
algorithm on two examples---a simple two-dimensional integrator and a space
vehicle rendezvous-docking problem---in Section \ref{sec:examples} and
provide conclusions and directions of future work in Section
\ref{sec:conlusion}.

%% file: problem-statement.tex
The following notation will be used throughout the paper:
we denote discrete-time time intervals by $\mathbb{Z}_{[a,b]} = \mathbb{Z} \cap
\{a, a+1, \dots, b-1,b\}$ for $a,b \in \mathbb{Z}$; the set of natural
numbers (including zero) as $ \mathbb{N}$; the Minkowski sum of
two sets $ \mathcal{S}_1, \mathcal{S}_2$ as $ \mathcal{S}_1 \oplus
\mathcal{S}_2=\{s_1+s_2:s_1\in \mathcal{S}_1, s_2 \in \mathcal{S}_2\}$; 
the Minkowski difference (or Pontryagin difference) of
two sets $ \mathcal{S}_1, \mathcal{S}_2$ as $ \mathcal{S}_2 \minkdiff
\mathcal{S}_1=\{s:s+s_1\in \mathcal{S}_2\ \forall s_1 \in \mathcal{S}_1\}$; 
and the indicator function corresponding to a set $ \mathcal{S}$ as $\ONE_{
\mathcal{S}}: \mathcal{X} \rightarrow \{0,1\}$ where $\ONE_{\mathcal{S}}(x)=1$
if $x\in \mathcal{S}$ and is zero otherwise.

\subsection{System formulation}

We consider a discrete-time, nonlinear, time-invariant system with an affine
disturbance,
\begin{equation}
  x_{k+1} = f(x_{k}, u_{k}) + w_{k}
  \label{eq:nonlin-dist}
\end{equation}
with state $x_{k} \in \mathcal{X} \subseteq \mathbb{R}^{n}$, input 
$u_{k} \in \mathcal{U} \subseteq \mathbb{R}^{m}$, disturbance $w_{k} \in \mathcal{W}\subseteq
\mathbb{R}^n$, and a function $f: \mathcal{X}\times \mathcal{U} \rightarrow
\mathcal{X}$. Without loss of generality, we assume $ \mathcal{W}$ contains
$0_{n}$, the zero vector of $ \mathbb{R}^{n}$. We will also consider the discrete, LTI system
of form
\begin{align}
    x_{k+1}&=Ax_k+Bu_k+w_k\label{eq:lin}
\end{align}
for some matrices $A\in \mathbb{R}^{n\times n}$ and $B\in \mathbb{R}^{n\times
m}$. We assume $A$ is non-singular, which holds true for discrete-time systems
that arise from the discretization of continuous-time systems.

\subsection{Robust reach-avoid sets}
\label{sub:RRA}

Let $t\in \mathbb{N}$ and $ \mathcal{F}$ denote the set of admissible state-feedback laws, $\nu:
\mathcal{X} \rightarrow \mathcal{U}$.  We define a control policy as a sequence
of state-feedback laws, $\rho_t = [\nu_{0}(\cdot),\dots,\nu_{t-1}(\cdot)]$ with $\nu_{k}\in
\mathcal{F}$ for $k\in \mathbb{Z}_{[0,t-1]}$. We denote the corresponding set of
admissible control policies as $\mathcal{P}_t$.

Let $\boundeddist\subseteq\mathcal{W}$ be a subset of the disturbance set. We define the $t$-time \emph{robust reach-avoid set} corresponding to $
\boundeddist$ as the set of initial states $x_{0} \in \mathcal{X}$ such that
there exists an admissible control policy $\rho_t\in \mathcal{P}_t$ that ensures
$x_{k}$ remains in a safe set $\mathcal{K}\subseteq \mathcal{X}$ for $k \in
\mathbb{Z}_{[0,t-1]}$ and $x_{t}$ lies in a target set $\mathcal{T}\subseteq
\mathcal{X}$ (reach-avoid objective) despite the presence of the disturbance $w_k\in \boundeddist$ at
each instant. Denoting $\bar{w}_t = {[w_{0}^\top, \dots,
w_{t-1}^\top]}^\top\in\mathcal{W}^{t}$, the $t$-time robust reach-avoid set is
\begin{align} 
  \rad_{t}(\mathcal{T},\mathcal{K},\boundeddist) &= \left\{ x_{0} \in \mathcal{X} :
    \exists \rho_t\in \mathcal{P}_t, \forall\bar{w}_t\in \boundeddist^{t}, \right.\nonumber \\
    &\quad\quad\quad\left.\forall k \in \mathbb{Z}_{[0,t-1]}, x_{k} \in
    \mathcal{K}, x_{t}
\in \mathcal{T} \right\}.
  \label{eq:dra-set}
\end{align}

Note that for $\boundeddist=\{0_{n}\}$, the system \eqref{eq:nonlin-dist} is equivalent to a
deterministic, discrete-time, nonlinear system
\begin{align}
    x_{k+1}&=f(x_k,u_k)\label{eq:nonlin}
\end{align}
when $w_k\in \boundeddist$.
The $t$-time \emph{viable set} of the system \eqref{eq:nonlin} is the set of
initial states $x_{0} \in \mathcal{X}$ such that there exists an admissible
control policy $\rho_t\in \mathcal{P}_t$ such that $x_{k}$ remains in a safe set
$\mathcal{K}$ for $k \in \mathbb{Z}_{[0,t]}$. That is,
\begin{align}
  \viab_{t}(\mathcal{K}) &= \left\{ x_{0} \in \mathcal{X} :
    \exists \rho_t\in \mathcal{P}_t,\forall k \in \mathbb{Z}_{[0,t]}, x_{k} \in \mathcal{K}\right\}\label{eq:viab-set}\\
    &= \rad_{t}( \mathcal{K}, \mathcal{K},\{0_n\}).\label{eq:viab-set-ra}
\end{align}
The authors in~\cite{maidens_2013} presented a Lagrangian formulation to compute
$\viab_t( \mathcal{K})$ and discussed the scalability of the viability analysis using MPT, ET, and support functions.

\subsection{Stochastic reach-avoid level-sets}
\label{sub:SRA}

In this subsection, we further assume the disturbance $w_k$ in
\eqref{eq:nonlin-dist} is an $n$-dimensional random vector defined in the
probability space $(\mathcal{W}, \sigma(\mathcal{W}),\mathbb{P}_{w})$. Here,
$\sigma( \mathcal{W})$ denotes the minimal $\sigma$-algebra associated with the
random vector $w_k$. We assume the disturbance $w_k$ is absolutely
continuous with a probability density function (PDF) $\psi_w$, the disturbance
process ${\{w_k\}}_{k=0}^{N-1}$ is an independent and identically distributed (i.i.d.)
random process, and $N\in \mathbb{N}$ is a finite time horizon. We assume
that $f$ is Borel-measurable, $ \mathcal{U}$ is compact, the sets $ \mathcal{K},
\mathcal{T}$ are Borel, and $\psi_w$ is continuous.

We denote the set of universally measurable state-feedback laws $\mu(\cdot):
\mathcal{X} \rightarrow \mathcal{U}$ as $\mathcal{F}_u$. We define the Markov
control policy as $\pi = [\mu_{0}(\cdot),\dots,\mu_{N-1}(\cdot)]$ where $\mu_{k}
\in \mathcal{F}_u\ \forall k\in \mathbb{Z}_{[0,N-1]}$, and $\mathcal{M}$ is the
set of admissible Markov policies. Since no measurability restrictions were
imposed on the feedback laws in Section~\ref{sub:RRA}, $ \mathcal{F}_u
\subseteq \mathcal{F}$ and $ \mathcal{M}\subseteq \mathcal{P}_N$.

Given a Markov policy $\pi$ and initial state $x_0\in \mathcal{X}$, the
concatenated state vector $\bar{x} = [ x_{1}, \dots, x_{N} ]$ for the system
\eqref{eq:nonlin-dist} is a random vector defined in the probability space
$(\mathcal{X}^{N},\sigma(\mathcal{X}^{N}), \mathbb{P}^{N,\pi}_{\bar{x}})$. The
probability measure $\mathbb{P}^{N,\pi}_{\bar{x}}$ is induced from the
probability measure $ \mathbb{P}_w$ via \eqref{eq:nonlin-dist}~\cite{summers2010_verification}.
We will denote the probability space associated with the random vector
$\bar{x}_{k} = [ x_{k+1}, \dots, x_{N} ]$ as $(\mathcal{X}^{N-k},
\sigma(\mathcal{X}^{N-k}), \mathbb{P}^{N-k,\pi}_{\bar{x}_{k}})$ for $k\in
\mathbb{Z}_{[0,N-1]}$.

For stochastic reachability analysis, we are interested in the maximum
likelihood that the system \eqref{eq:nonlin-dist} starting at an initial state
$x_0\in \mathcal{X}$ will achieve the reach-avoid objective using a Markov
policy. The maximum likelihood and the optimal Markov policy can be determined
as the solution to the optimization problem,
\cite{summers2010_verification}
\begin{equation} \label{eq:ra-value-fcn}
	\sup_{\pi\in \mathcal{M}} \mathbb{E}^{N,\pi}_{\bar{x}}\left[
    \left(\prod_{i=0}^{N-1} \ONE_{\mathcal{K}}(x_{i}) \right) \ONE_{\mathcal{T}}(x_{N})\right].
\end{equation}
A dynamic programming approach was presented in~\cite{summers2010_verification}
to solve problem \eqref{eq:ra-value-fcn}. Let the optimal solution to problem
\eqref{eq:ra-value-fcn} be $\pi^\ast=[\mu_0^\ast(\cdot)\ \ldots\
\mu_{N-1}^\ast(\cdot)]$, the \emph{maximal Markov policy in the terminal
sense}~\cite[Def. 10]{summers2010_verification}. The existence of a Markov policy is guaranteed for a continuous $\psi_w$ and compact $ \mathcal{U}$~\cite[Thm. 1]{VinodArxiv2017}. The approach
in~\cite{summers2010_verification} generates value functions $V_{k}^{\ast}: \mathcal{X}
\rightarrow [0,1]$ for $k\in[0,N]$,
\begin{align}
    V_{k}^{\ast}(x)&= \ONE_{\mathcal{K}}(x)\int_{ \mathcal{X}} V_{k+1}^{\ast}(y)\psi_{w}(y-f(x,\mu_k^\ast(x)))dy \nonumber\\
        &=\ONE_{ \mathcal{K}}(x)\mathbb{P}^{N-k,\pi^\ast}_{\bar{x}_{k}}\left(x_{N}
\in\mathcal{T}, x_{N-1} \in\mathcal{K}, \dots,\right. \nonumber \\
        &\quad\quad\quad\quad\quad\quad\quad\quad\quad\quad\quad\quad\quad\quad\left. x_{k+1}\in\mathcal{K}\vert x \right)\label{eq:valFun}
\end{align}
initialized with
\begin{align}
    V_N^{\ast}(x)&= \ONE_{ \mathcal{T}}(x). \label{eq:valFunN}
\end{align}
By definition, the optimal value function $V_{0}^{\ast}(x_0)$ provides the
maximum likelihood, optimal value of problem \eqref{eq:ra-value-fcn}, of
achieving the reach-avoid objective by the system \eqref{eq:nonlin-dist} for the
time horizon $N$ and the initial state $x_0\in \mathcal{X}$.

For $\beta\in[0,1]$ and $k\in \mathbb{Z}_{[0,N]}$, the \emph{stochastic reach-avoid $\beta$-level-set},% at time $k$,
\begin{equation} \label{eq:sra-set}
	\mathcal{L}_{k}(\beta) = \left\{ x\in\mathcal{X} : V_{N-k}^{\ast}(x) \geq
    \beta \right\},
\end{equation}
is the set of states $x$ that achieve the reach-avoid objective by the time
horizon with a probability of, at minimum, $\beta$, in the time interval $
\mathbb{Z}_{[0,k]}$.

\subsection{Problem statements}

The following problems are addressed in this paper:

\begin{prob}
    Construct a recursion for exact computation of the robust reach-avoid sets
    \eqref{eq:dra-set} for the system \eqref{eq:nonlin-dist}.\label{prob:RA}
\end{prob}
\begin{prob}
    Given a set $\boundeddist\subseteq \mathcal{W}$
    and the corresponding robust reach-avoid set \eqref{eq:dra-set},
    characterize the sufficient conditions under which there exists an optimal control policy that
    is a Markov control policy 
    for the system \eqref{eq:nonlin-dist}.
    \label{prob:meas}
\end{prob}
\begin{prob}
     Given $\beta\in[0,1]$, characterize $\boundeddist \subseteq \mathcal{W}$ whose
     corresponding robust reach-avoid set \eqref{eq:dra-set}
     underapproximates the stochastic reach-avoid $\beta$-level-set
     \eqref{eq:sra-set}.\label{prob:SRA}
\end{prob}
\begin{customthm}{3a}
    For a given $\beta \in [0,1]$, characterize an algorithm to compute $\boundeddist$ for
    Problem~\ref{prob:SRA} when the disturbance in \eqref{eq:nonlin-dist} is Gaussian.\label{prob:Gaussian}
\end{customthm}

%% file: lagrangian.tex
In this section, we characterize the robust reach-avoid set for the system
described in (\ref{eq:nonlin-dist}). To solve Problem~\ref{prob:RA}, we first
extend the approach presented in~\cite{maidens_2013,saintpierre1994_viability}
to reproduce the results presented in~\cite{bertsekas1971minimax}. The authors
in~\cite{maidens_2013} demonstrated scalability of the Lagrangian methods for
viability analysis in deterministic systems. By unifying these approaches, we
aim for a tractable and efficient Lagrangian computation of the robust
reach-avoid set with established scalability properties. We also demonstrate that the recursion presented for
the viable set computation in deterministic system~\cite{maidens_2013} is a
special case of the proposed Lagrangian approach. Finally, we solve
Problem~\ref{prob:meas} and establish that there is an optimal control policy
for the robust reach-avoid set that is also a Markov policy.

\subsection{Iterative computation for robust reach-avoid sets}

Similar to the work in~\cite{maidens_2013}, for the system
\eqref{eq:nonlin-dist}, we define the unperturbed, one-step
forward reach set from a point $x\in \mathcal{X}$ as $\reachf(x)$, and the
unperturbed, one-step backward reach set from a set $ \mathcal{S} \subseteq \mathcal{X}$ as $\reach(
\mathcal{S})$. Formally, for the system \eqref{eq:nonlin-dist},
\begin{align}
    \reachf(x) &\triangleq \{ x^+\in \mathcal{X} : u \in \mathcal{U},\ x^+=f(x,u)\}\label{eq:reachF}\\
    \reachd(\mathcal{S})&\triangleq\left\{ x^- \in \mathcal{X} : \exists u\in
\mathcal{U}, \exists y\in \mathcal{S},\ y=f(x^-,u)\right\} \nonumber \\ %\label{eq:dr1-set}\\
&=\left\{ x^- \in \mathcal{X} : \reachf(x^-)\cap \mathcal{S}\neq\emptyset
\right\} \label{eq:dFcapS}
\end{align}
where \eqref{eq:dFcapS} follows from \eqref{eq:reachF}.  For the system \eqref{eq:lin}, 
\begin{align}
    \reachf(x)&=A\{x\}\oplus B \mathcal{U},\label{eq:FOlin} \\
    \reachd( \mathcal{S})&= A^{-1}(\mathcal{S}\oplus (-B
    \mathcal{U})).\label{eq:reachOlinCts}
\end{align}
\begin{prop}
    Given a set $\boundeddist\subseteq\mathcal{W}$, the finite horizon robust
    reach-avoid sets for the system \eqref{eq:nonlin-dist} can be computed
    recursively as follows for $k\geq 1,k\in \mathbb{N}$:\label{lem:dra-sp-recursion} 
        \begin{align}
        \rad_{0}(\mathcal{T},\mathcal{K},\boundeddist) &= \mathcal{T} \label{eq:dra_recurs_lem1}\\
        \rad_{k}(\mathcal{T},\mathcal{K},\boundeddist) &= \left\{x_{0} \in
            \mathcal{K} : \right. \nonumber\\
        &\quad \left. \reachf(x_{0}) \cap (\rad_{k-1}( \mathcal{T}, \mathcal{K}) \minkdiff \boundeddist) \neq
        \emptyset \right\}.\label{eq:dra_recurs_lem2}
    \end{align}
\end{prop}
\begin{proof} 
    We first show the case $k=1$, which differs slightly
    from other cases. From
    \eqref{eq:nonlin-dist} and \eqref{eq:dra-set},
  \begin{align*}
    \rad_{1}(\mathcal{T},\mathcal{K},\boundeddist) &= \left\{x_{0} \in
      \mathcal{X} :  \exists \nu_0(\cdot)\in \mathcal{F}, x_{0} \in \mathcal{K},
      \forall w_{0} \in \boundeddist, \right. \\
   & \left. \hskip0.6cm \exists x^+\in \mathcal{T}, x^+=f(x_0,\nu_0(x_0))+w_0\right\} \\
    &= \left\{x_{0} \in
      \mathcal{X} : x_{0} \in \mathcal{K}, \exists u \in \mathcal{U}, \exists
      \nu_0(\cdot)\in \mathcal{F}, \right.  \nonumber \\
  &\quad \exists y\in \left.(\mathcal{T}\minkdiff \boundeddist) , y=f(x_0,\nu_0(x_0)),u=\nu_0(x_0)\right\} \\
             &= \left\{x_{0} \in \mathcal{K} : \exists y, y \in \reachf(x_{0})
  \wedge y\in (\mathcal{T} \minkdiff \boundeddist) \right\} \\
             &= \left\{x_{0} \in \mathcal{K} : \reachf(x_{0}) \cap (\rad_{0}(\mathcal{T},\mathcal{K},\boundeddist) \minkdiff \boundeddist) \neq \emptyset \right\}
  \end{align*}
  For any $t\in \mathbb{N}, t>1$, from \eqref{eq:dra-set},
  \begin{align}
      \rad_{t-1}(\mathcal{T}, \mathcal{K},\boundeddist) &= \left\{ x_{0} \in \mathcal{X} :
      \exists \rho_{t-1}\in \mathcal{P}_{t-1},\forall \bar{w}_{t-1}\in
      \boundeddist^{t-1},\right. \nonumber \\ 
  & \quad\left.  \forall k \in \mathbb{Z}_{[0,t-2]}, x_{k} \in \mathcal{K}, x_{t-1} \in
  \mathcal{T} \right\}.\label{eq:dra-set-t}
  \end{align}
  Using \eqref{eq:dra-set-t}, we construct $\rad_{t}(\mathcal{T},
  \mathcal{K}, \boundeddist)$ in the form of \eqref{eq:dra_recurs_lem2}.
  \begin{align}
      &\rad_{t}(\mathcal{T},\mathcal{K},\boundeddist) \nonumber \\
      &= \left\{ x_{0} \in \mathcal{X} :
      \exists \rho_{t}\in \mathcal{P}_{t},\forall \bar{w}_{t}\in
      \boundeddist^{t}, \forall k \in \mathbb{Z}_{[0,t-1]}, \right.\nonumber \\ 
     &\quad\quad\quad\quad\left. x_{k} \in
  \mathcal{K}, x_{t} \in \mathcal{T} \right\} \nonumber \\
            &= \left\{ x_{0} \in \mathcal{X} : x_0\in \mathcal{K}, \exists
  \nu_0(\cdot)\in \mathcal{F}, \exists \rho_{t-1}\in \mathcal{P}_{t-1},
  \right. \nonumber \\
  &\quad\quad\quad\quad\rho_t=[\nu_0, \rho_{t-1}], \forall
  w_0\in \boundeddist,\forall \bar{w}_{t-1}\in \boundeddist^{t-1}, \nonumber \\
  &\quad\quad\quad\quad\left.  \forall k \in \mathbb{Z}_{[1,t-1]},x_{k}
        \in \mathcal{K}, x_{t} \in \mathcal{T} \right\} \label{eq:prevP1}\\
            &= \left\{ x_{0} \in \mathcal{K} : \exists \nu_{0}(\cdot)\in \mathcal{F},\
  \forall w_0\in \boundeddist, \exists x_0^+\in \rad_{t-1}( \mathcal{T},
  \mathcal{K}, \boundeddist)\right. \nonumber \\
            &\left.\quad\quad\quad\quad f(x_0,\nu_0(x_0))+w_0 = x_0^+  \right\}
  \label{eq:nextP1}\\
            &= \left\{ x_{0} \in \mathcal{K} : \exists \nu_{0}(\cdot)\in
  \mathcal{F},\exists y\in \rad_{t-1}( \mathcal{T}, \mathcal{K})\minkdiff
    \boundeddist, \right. \nonumber \\
&\quad\quad\quad\quad\left. y= f(x_0,\nu_0(x_0)) \right\} \nonumber \\
            &= \left\{x_{0} \in \mathcal{K} : \reachf(x_{0}) \cap
(\rad_{t-1}(\mathcal{T},\mathcal{K},\boundeddist) \minkdiff \boundeddist) \neq
\emptyset \right\}. \nonumber
  \end{align}
Since the choice of $w_0$ depends only $(x_0,\nu(x_0))$, the terms $\exists
\rho_{t-1}$ and $\forall w_0$ can be exchanged in \eqref{eq:prevP1}. We obtain
\eqref{eq:nextP1} after exchanging the terms and applying \eqref{eq:dra-set-t}.
\end{proof}

\begin{figure*}
  \begin{center}
    \includegraphics{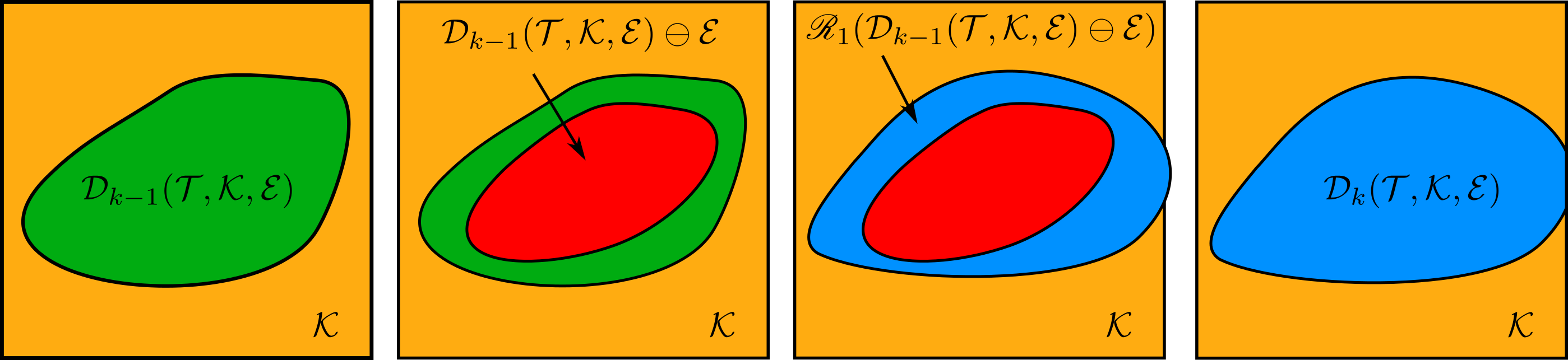}
  \end{center}
  \caption{Graphical representation of Lagrangian methods for computing 
  $\rad_{k}(\mathcal{T},\mathcal{K},\boundeddist)$ from
  $\rad_{k-1}(\mathcal{T},\mathcal{K},\boundeddist)$ via \eqref{eq:dra_recursen}.}
  \label{fig:ra-bound-dist}
\end{figure*}
\begin{theorem}\label{thm:DRA_recurs}
  For the system given in (\ref{eq:nonlin-dist}), the finite-time robust
  reach-avoid sets $\rad_{k}$ can be computed using the recursion  for $k\geq 1,k\in \mathbb{N}$:
  \begin{align}
      \rad_{0}(\mathcal{T},\mathcal{K},\boundeddist) &= \mathcal{T}
      \label{eq:dra_recurse0}\\
      \rad_{k}(\mathcal{T},\mathcal{K},\boundeddist) &= \mathcal{K} \cap
        \reachd(\rad_{k-1}( \mathcal{T}, \mathcal{K},\boundeddist) \minkdiff \boundeddist)\label{eq:dra_recursen}
  \end{align}
\end{theorem}
\begin{proof}
    Follows from Proposition~\ref{lem:dra-sp-recursion} and \eqref{eq:dFcapS}.
\end{proof}
Figure \ref{fig:ra-bound-dist} depicts the recursion in Theorem~\ref{thm:DRA_recurs} \eqref{eq:dra_recursen} graphically. From
\eqref{eq:dra_recursen}, we have the following corollary.
\begin{corollary}
    $\rad_k(\mathcal{T},\mathcal{K},\boundeddist) \subseteq K\ \forall k\geq 1,
    k\in \mathbb{N}$.\label{corr:Prop2}
\end{corollary}

For completeness, we establish that the viability analysis presented
in~\cite{maidens_2013} is a special case of Theorem~\ref{thm:DRA_recurs}.
From~\cite[Theorem 2.1]{kolmanovsky1998}, for any $ \mathcal{S}_1,
\mathcal{S}_2\subseteq \mathcal{X}$, $\mathcal{S}_2\minkdiff
\mathcal{S}_1=\cap_{s\in \mathcal{S}_1} ( \mathcal{S}_2 \oplus \{-s\})$.
Hence, 
\begin{align}
    \mathcal{S}_2\minkdiff \{0_n\}&=\mathcal{S}_2.\label{eq:S2diff0}
\end{align}
\begin{corollary}{\cite[Theorem 1]{maidens_2013}}
  The finite horizon viable sets for \eqref{eq:nonlin} can be computed
  recursively as follows:\label{corr:maidens}
  \begin{equation}
    \begin{split}
      \viab_{0}(\mathcal{K}) &= \mathcal{K} \\
      \viab_{k}(\mathcal{K}) &= \mathcal{K} \cap \reach(\viab_{k-1}(\mathcal{K}))
    \end{split}
  \end{equation}
\end{corollary}
\begin{proof}
    Follows from Theorem~\ref{thm:DRA_recurs}, \eqref{eq:nonlin}, and
    \eqref{eq:S2diff0}.
\end{proof}
A similar recursion can be provided for computing the reach-avoid sets for a
deterministic system.  

\begin{lemma}{~\cite[Proposition 3]{bertsekas1971minimax}}
    For the dynamics \eqref{eq:lin}, if $ \mathcal{U}, \mathcal{K}, \mathcal{T}$
    are convex and compact sets, $\boundeddist$ is a compact set, and $A$ in the
    dynamics \eqref{eq:lin} is non-singular, then $\rad_{k}( \mathcal{T},
    \mathcal{K}, \boundeddist)$ is convex and compact $\forall k \in
    \mathbb{N}$.\label{lem:dra-cvx}
\end{lemma}

Note that convexity of $\rad_{k}( \mathcal{T}, \mathcal{K}, \boundeddist)$ does
not require convexity of $\boundeddist$. Further, for polyhedral $ \mathcal{U},
\mathcal{K}, \mathcal{T}$, the robust reach-avoid set $\rad_k( \mathcal{T},
\mathcal{K}, \boundeddist)$ is polyhedral for $k\in \mathbb{N}$. Note that the
same can not be said be for ellipsoids~\cite[Sec.  4]{bertsekas1971minimax}.  A
detailed discussion for the implementation of Theorem~\ref{thm:DRA_recurs} for
polyhedral sets using support functions is given in~\cite[App.  A]{bertsekas1971minimax}.

\subsection{Minmax problem for robust reach-avoid set computation}

We will now address Problem~\ref{prob:meas}.  A minmax optimization
problem was presented in~\cite[Sec.  1]{bertsekas1971minimax},~\cite[Sec.
4.6.2]{bertsekasDP} to compute the robust reach-avoid sets \eqref{eq:dra-set} for
the system \eqref{eq:nonlin-dist}. The optimization problem is:
\begin{align}
  \begin{array}{rl}
      \underset{\rho_{t}}{\mbox{minimize}}&\hskip-0.3cm\underset{\bar{w}_{t}}{\mbox{
      maximize}}\  J(\rho_{t-1},\bar{w}_{t-1};x_0,t)=\sum_{k=0}^{t} g_k(x_k) \\
      \mbox{subject to}& \hskip-0.2cm\left\{\begin{array}{rll}
      x_{k+1} &= f(x_k,\nu_k(x_k))+w_k & k\in  \mathbb{Z}_{[0,t-1]}\\
      w_k & \in \boundeddist  & k\in  \mathbb{Z}_{[0,t-1]}\\
      \nu_k(\cdot) & \in \mathcal{F}  & k\in  \mathbb{Z}_{[0,t-1]}\\
    \end{array}\right.
  \end{array}
  \label{eq:prob-minmax}
\end{align}
where the decision variables are $ \rho_t$ and $\bar{w}_t$. Here,
$g_k(\cdot)=1-\ONE_{\mathcal{K}}(\cdot)$ for $k\in \mathbb{Z}_{[0,t-1]}$ and
$g_t(\cdot)=1-\ONE_{\mathcal{T}}(\cdot)$. The objective
function $J(\cdot)$ is parameterized by the initial state $x_0\in \mathcal{X}$
and the time horizon $t\in \mathbb{N}$. Problem \eqref{eq:prob-minmax} can be
solved using dynamic programming~\cite[Sec 1.6]{bertsekasDP} to generate the value
functions $J^\ast_k(x;t): \mathcal{X} \rightarrow \mathbb{Z}_{[0,t-k+1]}$ for
$k\in \mathbb{Z}_{[0,t]}$
\begin{align}
    H^\ast_{k}(u,x;t)&=\sup_{w\in \boundeddist}
    \left[J^\ast_{k+1}(f(x,u)+w;t)+g_k(x)\right]\label{eq:minmaxVH}\\
    J^\ast_k(x;t)&=\inf_{u\in \mathcal{U}}H^\ast_{k}(u,x;t) \label{eq:minmaxV}
\end{align}
initialized with $J^\ast_t(x;t)=g_t(x)$. The optimal value of problem \eqref{eq:prob-minmax} when starting at $x_0$ is $J^\ast_0(x_0;t)$. 
Further,
\begin{align}
    \rad_t( \mathcal{T}, \mathcal{K}, \boundeddist)&=\{x\in \mathcal{X}:
J^\ast_t(x_0;t)=0\}.\label{eq:link}
\end{align}

Recall that lower semi-continuous functions are functions whose sublevel-sets
are closed and upper-semicontinuous functions are functions whose negative is a lower semi-continuous function~\cite[Definition 7.13]{BertsekasSOC1978}. Also, the supremum of a 
lower-semicontinuous function is the negative of the infimum of an upper-semicontinuous function.
Let an optimal control policy for problem \eqref{eq:prob-minmax} be $\rho_t^\ast=[\nu_0^\ast(\cdot)\
\ldots\ \nu_{t-1}^\ast(\cdot)]$. Note that $\rho_t^\ast$ need not be unique. 

\begin{theorem}
   For closed sets $ \mathcal{K}, \mathcal{T}$ and compact set $ \mathcal{U}$, there exists an optimal policy $\rho^\ast_t$ for problem \eqref{eq:prob-minmax} which is also a Markov policy. \label{thm:LSCB}
\end{theorem}
\begin{proof}
    We show by induction that the statement \emph{S}: the optimal value
    functions $J_k^\ast$ of \eqref{eq:prob-minmax} are lower-semicontinuous and
    there exists a Borel-measurable state-feedback law $\nu_k^\ast(\cdot)$ for every $k\in \mathbb{Z}_{[0,N-1]}$.  Since Borel-measurability
    implies universal measurability~\cite[Definition 7.20]{BertsekasSOC1978},
    the proof of Theorem~\ref{thm:LSCB} follows from \emph{S} and the
    definition of a Markov policy. 
    
    \emph{Proof of S}: The closedness property of $ \mathcal{K}, \mathcal{T}$ imply
    $g_k(\cdot)$ is lower semi-continuous for $k\in \mathbb{Z}_{[0,t]}$. Hence,
    $J_t^\ast(x;t)$ is lower-semicontinuous. 
    
    Consider the base case $k=t-1$. From~\cite[Prop. 7.32 (b)]{BertsekasSOC1978}, we can
    see that $H_{t-1}^\ast(u,x;t)$ is lower semi-continuous. From~\cite[Prop.
    7.33]{BertsekasSOC1978}, we conclude that $J_{t-1}^\ast(x;t)$ is lower
    semi-continuous and an optimal state-feedback policy $\nu_{t-1}^\ast(\cdot)$ exists which is also Borel-measurable.

    Let $\tau\in \mathbb{Z}_{[1,t-2]}$. Assume, for induction, the case $k=\tau$ is true,
    i.e, $J_\tau^\ast$ is lower semi-continuous. The proof that $J_{\tau-1}^\ast$ is
    lower semi-continuous and the existence of a Borel-measurable $\nu_{\tau-1}^\ast(\cdot)$ follows
    from~\cite[Prop.  7.32(b) and 7.33]{BertsekasSOC1978}. This completes the induction.
\end{proof}

%% file: stochastic-approx.tex
We will now focus on the stochastic system described in Section~\ref{sub:SRA}
and use the theory developed in Section~\ref{sec:lagrangian} to solve
Problems~\ref{prob:SRA} and~\ref{prob:Gaussian}.

\begin{theorem}
   Given closed sets $ \mathcal{K}, \mathcal{T}$, compact set $ \mathcal{U}$ and
   a set $ \boundeddist \subseteq \mathcal{W}$. For every $x\in \rad_{t}(
   \mathcal{T}, \mathcal{K}, \boundeddist)$ with $t\in \mathbb{Z}_{[1,N]}$,
   \begin{align}
      \mathbb{P}^{t,\rho_N^\ast}_{\bar{x}_{t}}(x_{N} \in \mathcal{T},
      &x_{N-1} \in \mathcal{K}, \dots, \nonumber \\
      &x_{N-t+1} \in \mathcal{K} | x, \bar{w}_{t}\in \boundeddist^{t}) = 1.\label{eq:Prop1}
  \end{align}\label{thm:Prop1}
\end{theorem}
\begin{proof}
    Follows from Theorem~\ref{thm:LSCB} and the definition of $\rad_{t}( \mathcal{T},
    \mathcal{K}, \boundeddist)$ \eqref{eq:dra-set}.
\end{proof}
\begin{theorem} \label{lem:bounded-conservative}
    Given $\beta\in[0,1]$, closed sets $ \mathcal{K}, \mathcal{T}$, and a compact set
    $ \mathcal{U}$, if for any $t\in\mathbb{Z}_{[0,N]}$, $\boundeddist
    \subseteq \mathcal{W}$ such that $\mathbb{P}_{w}(w_k \in \boundeddist) =
    \beta^{\frac{1}{t}}$ for all $k\in\mathbb{Z}_{[0,t-1]}$, then $\rad_{t}(\mathcal{T},\mathcal{K},\boundeddist) \subseteq
    \mathcal{L}_{t}(\beta)$  .\label{thm:underapprox}
\end{theorem}
\begin{proof}
	The case for $t=0$ follows trivially from \eqref{eq:valFunN}, \eqref{eq:sra-set}, and \eqref{eq:dra_recurs_lem1}. Let $t>0$ and $x \in \rad_{t}(
    \mathcal{T}, \mathcal{K}, \boundeddist)$. 
    We are interested in underapproximating
    $\mathcal{L}_{t}(\beta)=\{x:V_{N-t}^\ast(x)\geq \beta\}$ as defined in~\eqref{eq:sra-set}. From \eqref{eq:valFun}, 
	\begin{align}
        &V_{N-t}^\ast(x)\nonumber
        \\
        &= \mathbb{P}^{t,\pi^\ast}_{\bar{x}_{t}}(x_{N} \in \mathcal{T}, x_{N-1}
        \in \mathcal{K}, \dots, x_{N-t+1} \in \mathcal{K} | x)\ONE_{ \mathcal{K}}(x) \nonumber \\
                 &= \mathbb{P}^{t,\pi^\ast}_{\bar{x}_{t}}(x_{N} \in \mathcal{T}, x_{N-1}
        \in \mathcal{K}, \dots,\nonumber \\
        &\ \quad\quad\quad\quad\quad\quad\quad\quad\quad\quad\quad\quad x_{N-t+1} \in \mathcal{K} | x, \bar{w}_{t}\in \boundeddist^{t}) \nonumber\\
        & \hskip2.5cm \times \mathbb{P}_{\bar{w}_{t}}^{t}(\bar{w}_{t}\in \boundeddist^{t}) \nonumber \\
                 &\quad+\mathbb{P}^{t,\pi^\ast}_{\bar{x}_{t}}(x_{N}
        \in \mathcal{T}, x_{N-1} \in \mathcal{K}, \dots, \nonumber \\
        &\ \quad\quad\quad\quad\quad\quad\quad\quad\quad x_{N-t+1} \in \mathcal{K} | x, \bar{w}_{t}\in (\mathcal{W}^{t}\setminus\boundeddist^{t})) \nonumber \\
        &\hskip2.5cm\times \mathbb{P}_{\bar{w}_{t}}^{t}(\bar{w}_{t}\in
        (\mathcal{W}^{t}\setminus\boundeddist^{t}))\label{eq:totProb}\\
        &\geq  \mathbb{P}^{t,\pi^\ast}_{\bar{x}_{t}}(x_{N} \in \mathcal{T}, x_{N-1}
        \in \mathcal{K}, \dots, \nonumber \\
        &\ \quad\quad\quad\quad\quad\quad\quad\quad\quad\quad\quad\quad x_{N-t+1} \in \mathcal{K} | x, \bar{w}_{t}\in \boundeddist^{t}) \nonumber\\
        & \hskip2.5cm \times \mathbb{P}_{\bar{w}_{t}}^{t}(\bar{w}_{t}\in
        \boundeddist^{t}).\label{eq:LowerBound} 
    \end{align}
    Equation \eqref{eq:totProb} follows from the law of total probability and
    Corollary~\ref{corr:Prop2} which implies $\ONE_{ \mathcal{K}}(x)=1$.
    Equation \eqref{eq:LowerBound}
    follows from \eqref{eq:totProb} after ignoring the second term (which is
    non-negative).
    Simplifying \eqref{eq:LowerBound} using Theorem~\ref{thm:Prop1} and the i.i.d.
    assumption of the disturbance process, we obtain
    \begin{align}
                   V_{N-t}^\ast(x) &\geq \mathbb{P}_{\bar{w}_{t}}^{k}(\bar{w}_{t} \in \boundeddist^{t})=\left(\mathbb{P}_{w}(w_{k} \in \boundeddist)\right)^{t} =
        \beta. \label{eq:approx}
	\end{align}
    Thus, $\rad_{t}(\mathcal{T},\mathcal{K},\boundeddist) \subseteq
    \mathcal{L}_{t}(\beta)$ by \eqref{eq:sra-set}.
\end{proof}

Theorem~\ref{lem:bounded-conservative} solves Problem~\ref{prob:SRA} for an
arbitrary density $\psi_{w}$. Computation of
$\rad_{t}(\mathcal{T},\mathcal{K},\boundeddist)$ can be done via
Theorem~\ref{thm:DRA_recurs}. Note that $\boundeddist$ characterized by
Theorem~\ref{lem:bounded-conservative} is not unique. Recall that
Corollary~\ref{corr:maidens} states that the robust reach-avoid set is exact
viable set for the case when $ \boundeddist=\{0_n\}$. We therefore prescribe
$\boundeddist$ that contains $\{0_n\}$ and has the least Lebesgue measure to reduce the
the degree of conservativeness in Theorem~\ref{lem:bounded-conservative}. 
We also recommend the set $\boundeddist$ be convex and compact for computational ease.

Next, we provide a method to compute $\boundeddist\subseteq \mathcal{W}$ for any $t\in\mathbb{Z}_{[0,N-1]}$ such that $\mathbb{P}_{w}(w_{k} \in \boundeddist) = \beta^{\frac{1}{t}}$ for all $k\in \mathbb{Z}_{[0,t-1]}$ when the disturbance in the system \eqref{eq:nonlin-dist} is a Gaussian random vector.

\subsection{Computation of $\boundeddist$ for Gaussian disturbance}
\label{sub:bound-dist-comp}

Let the disturbance in \eqref{eq:nonlin-dist} be $w_k=v$, an $n$-dimensional
Gaussian random variable with mean vector $\mu$ and covariance matrix $\Sigma$.
The probability density of a multivariate Gaussian random vector is~\cite[Ch.
29]{billingsley_probability_1995}
\begin{align}
    \psi_{v}(s)={(2\pi)}^\frac{-n}{2}{\vert\Sigma\vert}^\frac{-1}{2}\exp{\left(-\frac{{(s-\mu)}^\top\Sigma^{-1}{(s-\mu)}}{2}\right)}.
    \nonumber
\end{align}

Consider the $n$-dimensional ellipsoid parameterized by $R^2\in[0,\infty)$ for $\boundeddist$,
\begin{align}
    \boundeddist_{R^2}=\left\{s\in \mathbb{R}^{n}: {(s-\mu)}^\top\Sigma^{-1}{(s-\mu)}\leq
R^2\right\}.\label{eq:ellipsoid}
\end{align}
For $\mu=0,\Sigma=r^2 I_n$, we have $\boundeddist_{R^2}=\{w: w^\top w \leq r^2 R^2 \}$,
an $n$-dimensional hypersphere of radius $rR$. We aim to compute the parameter
$R^2$ such that $\mathbb{P}_v\{v\in \boundeddist_{R^2}\}=\beta^\frac{1}{t}$ for
application of Theorem~\ref{thm:underapprox}.

Given a standard normal distributed $n$-dimensional random vector $\eta\sim N(0,I_{n})$,
$v={\Sigma}^\frac{1}{2}\eta+\mu$~\cite[Ch.
29]{billingsley_probability_1995}.  Also, $\boundeddist_{R^2}={\Sigma}^\frac{1}{2}
\boundeddist_{\eta,R^2}\oplus\{\mu\}$ with $\boundeddist_{\eta,R^2}=\{s\in
\mathbb{R}^{n}:s^\top s\leq R^2\}$. Since the affine transformation of $\eta$ to
$v$ is deterministic, $\mathbb{P}_v\{v\in\boundeddist_{R^2}\} =\mathbb{P}_{\eta}\{\eta\in \boundeddist_{\eta,R^2}\}=\beta^\frac{1}{t}$.  From~\cite[Ex.
20.16]{billingsley_probability_1995}, we have
$$F_{\chi^2(n)}(R^2)=\mathbb{P}\left\{\chi^2(n)\leq
R^2\right\}=\mathbb{P}\{\eta\in \boundeddist_{\eta,R^2}\}=\beta^\frac{1}{t}$$
where $\chi^2(n)$ is a chi-squared random variable with $n$ degrees of freedom
and $F_{\chi^2(n)}(\cdot)$ denotes its cumulative distribution function.
Consequently, we have
\begin{align}
    R^2=F^{-1}_{\chi^2(n)}\left(\beta^\frac{1}{t}\right).\label{eq:R2}
\end{align}
Equations \eqref{eq:ellipsoid} and \eqref{eq:R2} solves Problem~\ref{prob:Gaussian}.

\subsection{Computing the stochastic level-set underapproximation}

A pseudo-algorithm to compute the underapproximation of the $N$-time stochastic reach-avoid $\beta$-level-set is 
shown in Algorithm \ref{alg:dra-algorithm} using robust reach-avoid sets. Note that while the system dynamics permitted for Algorithm~\ref{alg:dra-algorithm} is the nonlinear system given in \eqref{eq:nonlin-dist}, the computation of $\reachd(S)$ is accessible only for linear system \eqref{eq:lin} as defined in \eqref{eq:reachOlinCts}. Further, Lemma~\ref{lem:dra-cvx} guarantees convexity and compactness of the robust reach-avoid sets, allowing for easy representation, only for linear system dynamics.

\begin{algorithm}[]
    \SetKwInOut{Input}{Input}\SetKwInOut{Output}{Output}

    \Input{Safe set, $\mathcal{K}$; target set, $\mathcal{T}$; system dynamics \eqref{eq:nonlin-dist}, desired probability level $\beta \in [0,1]$, Gaussian covariance matrix and mean, $\Sigma$, $\mu$; and time horizon, $N$}
    \Output{$N$-time stochastic reach-avoid $\beta$-level-set underapproximation, $\rad_{N}(\mathcal{K},\mathcal{T},\boundeddist)$}
    \vskip6pt
    Compute $\boundeddist$ for $\beta^{\frac{1}{N}}$ \hfill // from (\ref{eq:ellipsoid}), (\ref{eq:R2}) \\
    $\rad_{0}(\mathcal{K},\mathcal{T},\boundeddist) \leftarrow \mathcal{T}$ \\
    \For{$i = 1, 2, \dots, N$}{
        $S \leftarrow \rad_{i-1}(\mathcal{K},\mathcal{T},\boundeddist) \minkdiff \boundeddist$  \hfill // from \eqref{eq:dra_recursen} \\
        $R \leftarrow \reach(S)$ \hfill // from \eqref{eq:dFcapS}\\
        $\rad_{i}(\mathcal{K},\mathcal{T},\boundeddist) \leftarrow \mathcal{K} \cap R$ \hfill // from \eqref{eq:dra_recursen}
    }

    \caption{Underapproximation of the $N$-time stochastic reach-avoid $\beta$-level-set for
    system (\ref{eq:nonlin-dist}) with a Gaussian disturbance.}
    \label{alg:dra-algorithm}
\end{algorithm}

Since these sets are formed using Lagrangian techniques, Algorithm~\ref{alg:dra-algorithm} is more computationally efficient than
the dynamic programming based discretization approach.
Algorithm~\ref{alg:dra-algorithm} requires a number of basic geometric
operations. We will focus on the implementation of Algorithm~\ref{alg:dra-algorithm} in a polyhedral representation using the readily available MATLAB toolbox MPT.
From Lemma~\ref{lem:dra-cvx}, we note that support function methods can also be used~\cite{maidens_2013,bertsekas1971minimax}.
The conservativeness of the underapproximations obtained using Algorithm \ref{alg:dra-algorithm}
are very problem dependent. The system dynamics, strength of the disturbance process,
and size of the targe and safe sets can all have non-trivial affects on the
resulting conservativeness.

The robust reach-avoid set computation requires a Minkowski difference
operation as well as an intersection operation in the recursion \eqref{eq:dra_recursen}.
Hence both facet and vertex representations are typically required for polytopes 
and numerical implementations will be limited by the well-known vertex-facet
enumeration problem. Support functions would not be subject to this problem but
require analytic solutions to support vector calculations.
Minkowski
differences can be handled for polytopes using the MPT toolbox \cite{MPT3} but
implementation using ellipses is not feasible without further underapproximation.
Additional problems such as redundancy in
vertices and facets also commonly arise using polytope representations.

%% file: examples.tex
All results were obtained using the MPT toolbox \cite{MPT3} with MATLAB
R2016a running on Windows 7 computer with
and Intel Core i7-2600 CPU, 3.6 GHz, and 8 GB RAM.
We focus on examples in which clear comparisons of conservativeness can be made 
since the ability to handle high-dimensional systems is established in
\cite{maidens_2013}.

\subsection{2-Dimensional Double Integrator}
\label{sec:double-int-ex}

The first example considered is the stochastic viability analysis of a 2-dimensional discrete-time
double integrator model. This example can be solved with both the proposed Lagrangian
methods as well as with dynamic programming, allowing for direct comparisons of
conservativeness and speed.

The discretized double integrator dynamics are
\begin{equation}
  x_{k+1} = \left[ \begin{array}{cc}
    1 & T \\
    0 & 1
  \end{array}\right] x_{k} + \left[\begin{array}{c}
    \frac{T^{2}}{2} \\
    T
  \end{array}\right] u_{k} + w_{k}
  \label{eq:disys}
\end{equation}
The state $x_{k} \in \mathcal{X}\subseteq\mathbb{R}^{2}$, input $u_{k} \in 
\mathcal{U} \subseteq \mathbb{R}$, $T = 0.25$, and the disturbance is assumed to be i.i.d.
Gaussian $w_{k} \sim N(0,0.005 I_{2})$.

\begin{figure*}
  \begin{center}
    \includegraphics{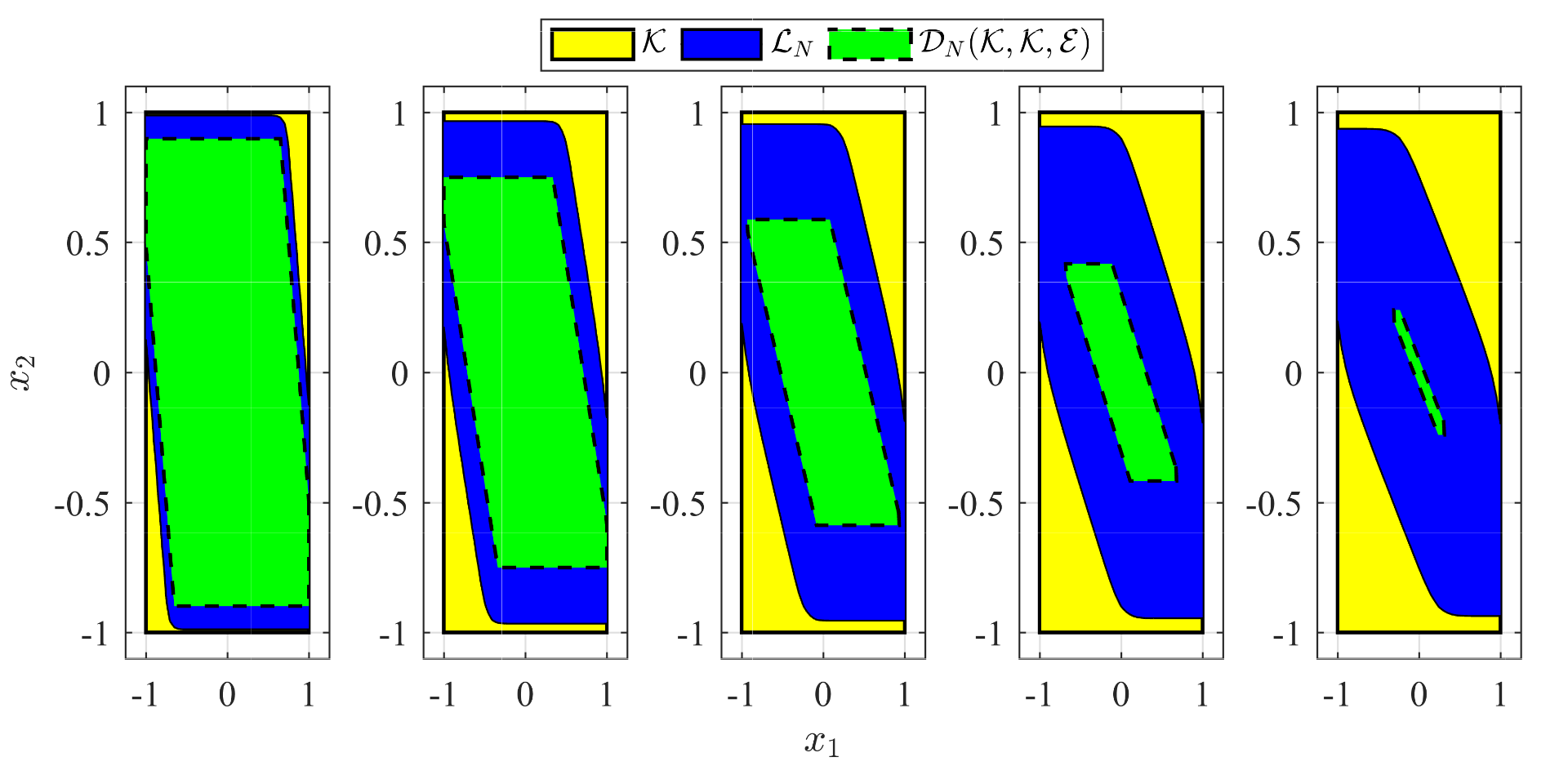}
  \end{center}
  \caption{Comparison between the dynamic programming $\beta=0.8$ level-sets and the
  		   Lagrangian underapproximation for a time horizon (from left) of $N=1,2,3,4,5$,
           for the double integrator system (\ref{eq:disys}) with a Gaussian
       disturbance $w_{k} \sim N(0, 0.005 I_2)$.}
  \label{fig:double-int}
\end{figure*}

Figure \ref{fig:double-int} compares the underapproximation via Algorithm 
\ref{alg:dra-algorithm}
and the level-sets computed using dynamic programming techniques, as in 
\cite{summers2010_verification}.The underapproximation
is closest and the approximation become progressively more conservative as $N$
increases, as is expected. For $N_{1}, N_{2} \in \mathbb{N}$, $N_{2} > N_{1}$,
$\beta^{\frac{1}{N_{2}}} > \beta^{\frac{1}{N_{1}}}$, and hence, from
Section \ref{sub:bound-dist-comp}, $R_{N_{2}} > R_{N_{1}}$, indicating that
$\boundeddist_{N_{1}} \subset \boundeddist_{N_{2}}$.

A comparison between the total computation time for the
dynamic programming method and Algorithm~\ref{alg:dra-algorithm} is provided in Table \ref{tab:di-simtimes}. The accuracy of dynamic programming relies on its grid size, resulting in a trade-off between accuracy and computation speed, from which Algorithm~\ref{alg:dra-algorithm} does not suffer. 

\begin{table}
  \begin{center}
    \begin{tabular}{|c|c|c|c|}
      \hline
      Grid Size & Dynamic Programming & Algorithm 1 & Ratio\\ \hline
      $41 \times 41$ & 8.16 & 0.98 & 8.3\\
      $82 \times 82$ & 59.76 & 0.98 & 60.9\\ \hline
    \end{tabular}
  \end{center}
  \caption{Computation times, in seconds,  for double integrator problem with dynamic programming
           and with Algorithm 1, and the ratio of computation times (dynamic programming by Algorithm 1). Algorithm 1 does not require a grid.}
  \label{tab:di-simtimes}
\end{table}

For systems with Gaussian disturbance processes that have a very low variance,
the underapproximation obtained through the Lagrangian methods tightly
approximates the stochastic level-set and is computed significantly faster---over $7$ times faster for a $41 \times 41$ grid. 
Figure~\ref{fig:dicomp-lowvar} shows a comparison of the stochastic level-set and
the Lagrangian underapproximation when the Gaussian disturbance process is
of the form $w_{k} \sim N(0, 10^{-5} I_{2})$. The bumps on the exterior
of the stochastic level-set are a numerical artifact from the state-space gridding.
\begin{figure}
    \includegraphics{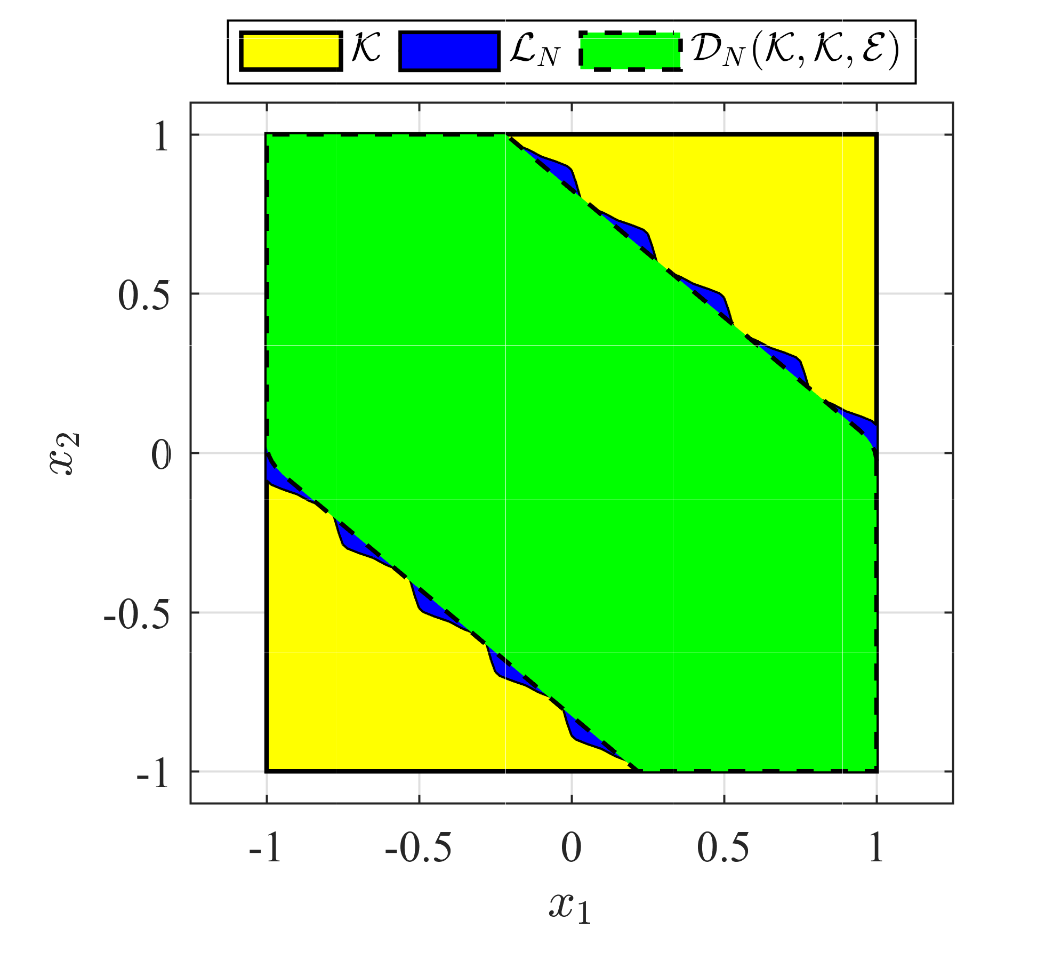}
  \caption{Comparison between the dynamic programming $\beta=0.8$ level-sets and
  the Lagrangian underapproximation at $N=5$, for the double integrator
  system (\ref{eq:disys}) with a Gaussian disturbance $w_{k} \sim N(0, 10^{-5} I_{2})$.}
    \label{fig:dicomp-lowvar}
\end{figure}

\subsection{Application to space-vehicle dynamics}

In this section, we compute an underapproximation of the stochastic
reach-avoid level-set for a spacecraft rendezvous docking problem using Algorithm~\ref{alg:dra-algorithm}. The goal is for a spacecraft, referred to as the deputy, to approach and dock
to an orbiting satellite, referred to as the chief, while remaining in a predefined line-of-sight
cone. The dynamics are described by the Clohessy-Wiltshire-Hill (CWH)
equations \cite{wiesel1989_spaceflight}
\begin{equation}
  \begin{split}
    \ddot{x} - 3 \omega x - 2 \omega \dot{y} = \frac{F_{x}}{m_{d}} \\
    \ddot{y} + 2 \omega \dot{x} = \frac{F_{y}}{m_{d}}
  \end{split}
  \label{eq:2d-cwh}
\end{equation}
The chief is located at the origin,
the position of the deputy is at $x,y \in \mathbb{R}$, $\omega = \sqrt{\mu/R_{0}^{3}}$
is the orbital frequency, $\mu$ is the gravitational constant, and $R_{0}$ is
the orbital radius of the spacecraft.

We define the state vector $z = [x,y,\dot{x},\dot{y}] \in \mathbb{R}^{4}$ and input
vector $u = [F_{x},F_{y}] \in \mathcal{U}\subseteq\mathbb{R}^{2}$. We discretize the dynamics (\ref{eq:2d-cwh}) in time to obtain the discrete-time LTI system,
\begin{equation}
  z_{k+1} = A z_{k} + B u_{k} + w_{k} \label{eq:lin-cwh}
\end{equation}
where $w_{k} \in \mathbb{R}^{4}$ is assumed to be a Gaussian i.i.d. disturbance with
$\mathbb{E}[w_{k}] = 0$, $\mathbb{E}[w_{k}w_{k}^\top] =
10^{-4}\times\mbox{diag}(1, 1, 0.0005, 0.0005)$.

We define the target set and the constraint set as in \cite{lesser2013_spacecraft}
\begin{align}
  \mathcal{T} &= \left\{ z \in \mathbb{R}^{4}: |z_{1}| \leq 0.1, -0.1 \leq z_{2} \leq 0, \right. \nonumber \\
              &\hskip1.9cm \left. |z_{3}| \leq 0.01, |z_{4}| \leq 0.01 \right\} \\
  \mathcal{K} &= \left\{z \in \mathbb{R}^{4}: |z_{1}| \leq z_{2}, |z_{3}| \leq 0.05, |z_{4}| \leq 0.05 \right\}\\
  \mathcal{U} &= [-0.1, 0.1] \times [-0.1, 0.1].
\end{align}

Figure \ref{fig:cwh-exampe} shows a cross-section at $\dot{x} = \dot{y} = 0$ of the
resulting 
underapproximation of the stochastic reach-avoid level set. The computation time for the 
$N=5$ level-set was 14.5 seconds. Because of the extreme computational requirements
of solving a 4-dimensional problem via dynamic programming methods we cannot make a direct computational comparison between dynamic programming and Algorithm 
\ref{alg:dra-algorithm}. In \cite[Figure 2]{lesser2013_spacecraft}, a
cross-section of $\dot{x} = \dot{y} = 0.9$ of the stochastic reach-avoid set was
approximated using convex, chance-constrained optimization and particle approximation
methods. Since both of these methods require gridding, the computation time is slower,
reported to be about 20 minutes, for a subset of the state space.
\begin{figure}
  \includegraphics{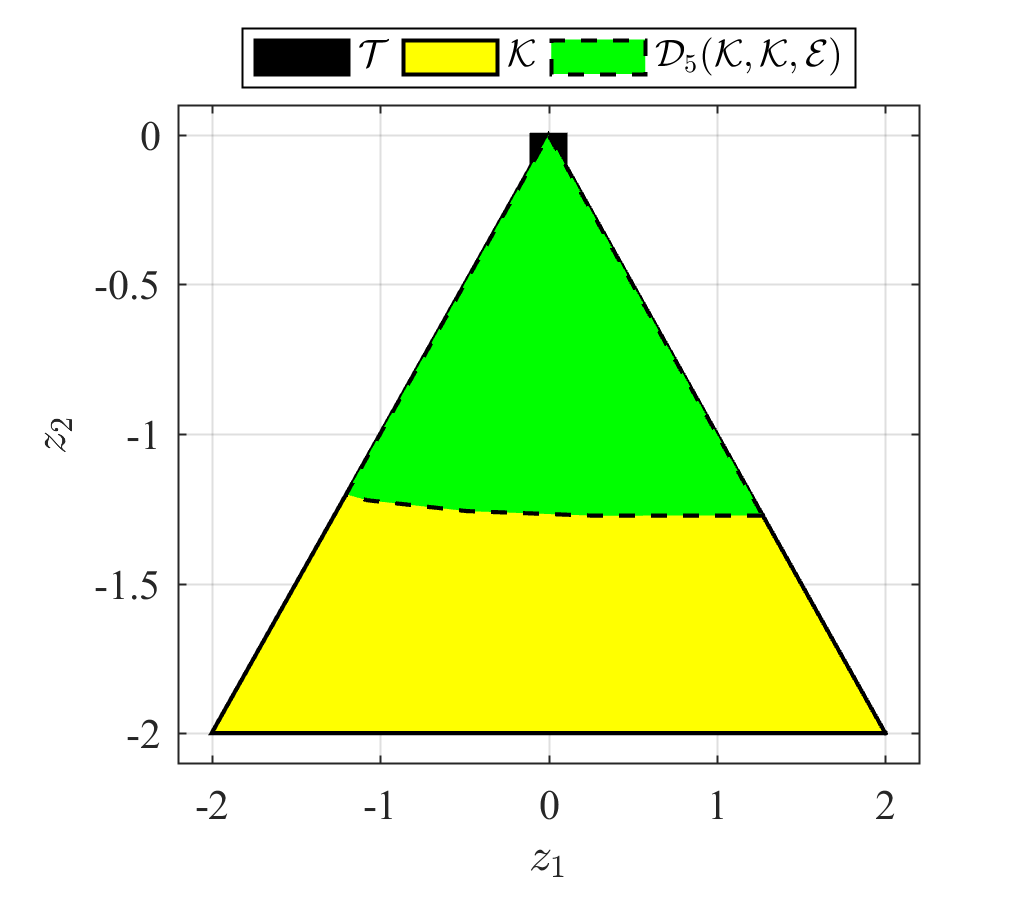}
  \caption{Cross-section of the underapproximation of the $N=5$ stochastic $\beta=0.8$ level-set for the
           spacecraft rendezvous docking problem (\ref{eq:lin-cwh}) at $\dot{x} = \dot{y} = 0$.}
  \label{fig:cwh-exampe}
\end{figure}

%% file: conclusion.tex
In this work, we provide a Lagrangian method to compute an underapproximation
of a stochastic reach-avoid level-set using robust reach-avoid sets. 
We synthesize approaches in~\cite{maidens_2013,saintpierre1994_viability} and~\cite{bertsekas1971minimax,kerrigan2001robust,rakovic2006_reach}, and characterize the sufficient conditions under which a optimal control policy
for the robust reach-avoid set is also a Markov policy. We demonstrate that our Lagrangian approach to compute the underapproximation is
significantly faster when compared to the dynamic programming approach.
The utility of this method is problem-dependent, as the conservativeness of the underapproximations
are affected by the system dynamics and noise processes.

In future, we intend to examine methods to reduce the conservativeness of the
underapproximation and extend the computation of the disturbance set
$\boundeddist$ in Theorem~\ref{thm:underapprox} for disturbances other than a
Gaussian random vector.

%% file: ms.bbl
\begin{thebibliography}{10}
\providecommand{\url}[1]{#1}
\csname url@samestyle\endcsname
\providecommand{\newblock}{\relax}
\providecommand{\bibinfo}[2]{#2}
\providecommand{\BIBentrySTDinterwordspacing}{\spaceskip=0pt\relax}
\providecommand{\BIBentryALTinterwordstretchfactor}{4}
\providecommand{\BIBentryALTinterwordspacing}{\spaceskip=\fontdimen2\font plus
\BIBentryALTinterwordstretchfactor\fontdimen3\font minus
  \fontdimen4\font\relax}
\providecommand{\BIBforeignlanguage}[2]{{%
\expandafter\ifx\csname l@#1\endcsname\relax
\typeout{** WARNING: IEEEtran.bst: No hyphenation pattern has been}%
\typeout{** loaded for the language `#1'. Using the pattern for}%
\typeout{** the default language instead.}%
\else
\language=\csname l@#1\endcsname
\fi
#2}}
\providecommand{\BIBdecl}{\relax}
\BIBdecl

\bibitem{lesser2013_spacecraft}
K.~Lesser, M.~Oishi, and R.~S. Erwin, ``Stochastic reachability for control of
  spacecraft relative motion,'' in \emph{Proc. IEEE Conf. on Decision and
  Ctrl.}, December 2013.

\bibitem{Tomlin2003}
C.~Tomlin, I.~Mitchell, A.~Bayen, and M.~Oishi, ``Computational techniques for
  the verification of hybrid systems,'' \emph{Proc. IEEE}, vol.~91, no.~7, pp.
  986--1001, 2003.

\bibitem{summers2011stochastic}
S.~Summers, M.~Kamgarpour, J.~Lygeros, and C.~Tomlin, ``A stochastic
  reach-avoid problem with random obstacles,'' in \emph{Proc. Hybrid Syst.:
  Comput. and Ctrl.}\hskip 1em plus 0.5em minus 0.4em\relax ACM, 2011, pp.
  251--260.

\bibitem{maidens_2013}
J.~N. Maidens, S.~Kaynama, I.~M. Mitchell, M.~M. Oishi, and G.~A. Dumont,
  ``Lagrangian methods for approximating the viability kernel in
  high-dimensional systems,'' \emph{Automatica}, vol.~49, no.~7, pp.
  2017--2029, 2013.

\bibitem{KariotoglouECC2011}
N.~Kariotoglou, D.~M. Raimondo, S.~Summers, and J.~Lygeros, ``A stochastic
  reachability framework for autonomous surveillance with pan-tilt-zoom
  cameras,'' in \emph{Proc. European Ctrl. Conf.}, 2011, pp. 1411--1416.

\bibitem{ManganiniCYB2015}
G.~Manganini, M.~Pirotta, M.~Restelli, L.~Piroddi, and M.~Prandini, ``Policy
  search for the optimal control of {Markov} {Decision} {Processes}: A novel
  particle-based iterative scheme,'' \emph{{IEEE} Trans. Cybern.}, pp. 1--13,
  2015.

\bibitem{summers2010_verification}
S.~Summers and J.~Lygeros, ``Verification of discrete time stochastic hybrid
  systems: A stochastic reach-avoid decision problem,'' \emph{Automatica},
  vol.~46, pp. 1951--1961, September 2010.

\bibitem{abate2008_reachability}
A.~Abate, M.~Prandini, J.~Lygeros, and S.~Sastry, ``Probabilistic reachability
  and safety for controlled discrete time stochastic hybrid systems,''
  \emph{Automatica}, vol.~44, pp. 2724--2734, October 2008.

\bibitem{AbateHSCC2007}
A.~Abate, S.~Amin, M.~Prandini, J.~Lygeros, and S.~Sastry, ``Computational
  approaches to reachability analysis of stochastic hybrid systems,'' in
  \emph{Proc. Hybrid Syst.: Comput. and Ctrl.}, 2007, pp. 4--17.

\bibitem{KariotoglouECC2013}
N.~Kariotoglou, S.~Summers, T.~Summers, M.~Kamgarpour, and J.~Lygeros,
  ``Approximate dynamic programming for stochastic reachability,'' in
  \emph{Proc. European Ctrl. Conf.}, 2013, pp. 584--589.

\bibitem{KariotoglouSCL2016}
N.~Kariotoglou, K.~Margellos, and J.~Lygeros, ``On the computational complexity
  and generalization properties of multi-stage and stage-wise coupled scenario
  programs,'' \emph{Syst. and Ctrl. Lett.}, vol.~94, pp. 63--69, 2016.

\bibitem{bertsekas1971minimax}
D.~P. Bertsekas and I.~B. Rhodes, ``On the minimax reachability of target sets
  and target tubes,'' \emph{Automatica}, vol.~7, no.~2, pp. 233--247, 1971.

\bibitem{kerrigan2001robust}
E.~C. Kerrigan, ``Robust constraint satisfaction: Invariant sets and predictive
  control,'' Ph.D. dissertation, University of Cambridge, 2001.

\bibitem{rakovic2006_reach}
S.~V. Rakovi\'{c}, E.~C. Kerrigan, D.~Q. Mayne, and J.~Lygeros, ``Reachability
  analysis of discrete-time systems with disturbances,'' \emph{{IEEE} Trans.
  Autom. Ctrl.}, vol.~51, no.~4, pp. 546--560, April 2006.

\bibitem{saintpierre1994_viability}
P.~Saint-Pierre, ``Approximation of the viability kernel,'' \emph{Applied
  Mathematics and Optimization}, vol.~29, no.~2, pp. 187--209, March 1994.

\bibitem{guernic2009}
C.~{Le Geurnic}, ``Reachability analysis of hybrid systems with linear
  continuous dynamics,'' Ph.D. dissertation, Universit\'{e} Joseph-Fourier,
  2009.

\bibitem{MPT3}
M.~Herceg, M.~Kvasnica, C.~Jones, and M.~Morari, ``{Multi-Parametric Toolbox
  3.0},'' in \emph{Proc.~of the European Control Conference}, Z\"urich,
  Switzerland, July 17--19 2013, pp. 502--510,
  \url{http://people.ee.ethz.ch/\%7Empt/3/}.

\bibitem{leguernic2010}
C.~{Le Guernic} and A.~Girard, ``Reachability analysis of linear systems using
  support functions,'' \emph{Nonlinear Analysis: Hybrid Systems}, vol.~4,
  no.~2, pp. 250--262, 2010.

\bibitem{ellipsoidal}
A.~A. Kurzhanskiy and P.~Varaiya, ``Ellipsoidal toolbox,'' University of
  California, Berkeley, Tech. Rep., 2006.

\bibitem{VinodArxiv2017}
\BIBentryALTinterwordspacing
A.~P. Vinod and M.~M.~K. Oishi, ``Scalable underapproximation for stochastic
  reach-avoid problem for high-dimensional {LTI} systems using {Fourier}
  transforms,'' in \emph{IEEE Control Systems Letters (L-CSS)}, 2017,
  (submitted). [Online]. Available: \url{https://arxiv.org/abs/1703.02135}
\BIBentrySTDinterwordspacing

\bibitem{kolmanovsky1998}
I.~Kolmanovsky and E.~G. Gilbert, ``Theory and computation of disturbance
  invariant sets for discrete-time linear systems,'' \emph{Mathematical
  Problems in Engineering}, vol.~4, pp. 317--367, 1998.

\bibitem{bertsekasDP}
D.~P. Bertsekas, \emph{Dynamic Programming and Optimal Control. {Vol}. 1},
  3rd~ed.\hskip 1em plus 0.5em minus 0.4em\relax Belmont, Mass: Athena
  Scientific, 2005.

\bibitem{BertsekasSOC1978}
D.~P. Bertsekas and S.~E. Shreve, \emph{Stochastic Optimal Control: the
  Discrete Time Case}.\hskip 1em plus 0.5em minus 0.4em\relax Academic Press,
  1978.

\bibitem{billingsley_probability_1995}
P.~Billingsley, \emph{Probability and Measure}, 3rd~ed.\hskip 1em plus 0.5em
  minus 0.4em\relax New York: Wiley, 1995.

\bibitem{wiesel1989_spaceflight}
W.~Wiesel, \emph{Spaceflight Dynamics}.\hskip 1em plus 0.5em minus 0.4em\relax
  New York: McGraw-Hill, 1989.

\end{thebibliography}
